\documentclass[journal]{IEEEtran}
\usepackage[nocompress]{cite}
\usepackage{amsmath,amssymb,amsfonts}
\usepackage{float}
\floatstyle{ruled}
\newfloat{algorithm}{t}{loa}
\floatname{algorithm}{Algorithm}
\usepackage{graphicx}
\usepackage{textcomp}
\usepackage{xcolor}
\usepackage{booktabs}
\usepackage{amsthm}
\makeatletter
\def\abstract{\normalfont\@IEEEabskeysecsize\noindent\textbf{\textit{\abstractname}: }\ignorespaces\@IEEEgobbleleadPARNLSP}

\def\IEEEkeywords{\normalfont\@IEEEabskeysecsize\vspace{0.4\baselineskip}\noindent\textbf{\textit{\IEEEkeywordsname}: }\ignorespaces\@IEEEgobbleleadPARNLSP}

\makeatother
\usepackage[hidelinks]{hyperref}
\usepackage{orcidlink}
\usepackage[capitalize]{cleveref}

\newtheorem{theorem}{Theorem}
\newtheorem{lemma}{Lemma}
\newtheorem{corollary}{Corollary}
\newtheorem{proposition}{Proposition}
\theoremstyle{definition}
\newtheorem{assumption}{Assumption}
\crefname{assumption}{Assumption}{Assumptions}
\Crefname{assumption}{Assumption}{Assumptions}
\crefname{algorithm}{Algorithm}{Algorithms}
\Crefname{algorithm}{Algorithm}{Algorithms}

\newcommand{\Skip}{\mathrm{Skip}}
\newcommand{\Rrisk}{R}

\begin{document}

\title{Audited Selective Verification for Risk-Controlled N-1 Thermal Contingency Screening under Deployment Shift}

\author{Jayakumar~Manoharan~\orcidlink{0009-0009-7765-9165},~\IEEEmembership{Senior~Member,~IEEE}%
\thanks{J. Manoharan is with the Electric Power Research Institute (EPRI), Charlotte, NC, USA (e-mail: jmanoharan@epri.com; ORCID: 0009-0009-7765-9165).}}

\markboth{IEEE Transactions on Power Systems}{Manoharan: Audited Selective Verification for N-1 Security}

\maketitle

\begin{abstract}
Real-time N-1 contingency screening in an energy management system trades assurance against cost: verifying every credible outage with full power flow is too slow, while fast linear-sensitivity screening gives no statistical guarantee and can silently pass unsafe operating points, especially when a controller drives the system into unfamiliar regimes. This paper introduces Audited Selective Verification, a risk-budgeted screening and triage layer for any controller's output (optimization, model-predictive, or learned). A cheap surrogate proposes which outages to skip; an online audit runs full power flow on a small random sample each window; and a calibrated threshold certifies a thermal-violation-rate bound for the skipped set at a chosen budget and confidence, with a corresponding bound for the unverified trusted subset. Validity rests on real verification and the audit rather than on surrogate accuracy, so it holds under arbitrary deployment shift. It is a risk-budgeted screen, not a replacement for deterministic verification when policy requires checking every credible contingency. On three public transmission systems up to 1354 buses, the realized violation rate stays within budget, standard deterministic and calibrated screens become unsafe under shift, and the method cuts full power-flow studies by 29 to 75 percent per real-time operating point.

\end{abstract}

\begin{IEEEkeywords}
N-1 security assessment, contingency analysis, transmission system operation, risk-based security assessment, distribution-free risk control.
\end{IEEEkeywords}

\section{Introduction}
\IEEEPARstart{K}{eeping} the transmission system secure against any single credible outage, the N-1 criterion, is a core task of real-time operation, and an energy management system must reassess it continuously as operating points change \cite{wood2013}. Verifying it exactly means a full AC power flow for every credible contingency, too slow to repeat after every candidate action, so operators rely on fast screening. This assurance requirement also gates adoption of the faster, more adaptive controllers now entering the control room, reinforcement-learning agents for topology control and model-predictive schemes for redispatch \cite{marot2021,lehna2024,saferlreview2024}: no operator will hand control to a policy without evidence that its actions respect N-1 security.

Providing that evidence cheaply is hard. The standard fast alternative is deterministic screening with linear sensitivities (power transfer and line outage distribution factors), which ranks contingencies and skips those it predicts to be safe \cite{wood2013}. Screening is fast, but it offers no statistical guarantee: it trusts a linear approximation, and when that approximation is wrong it silently passes unsafe actions. The problem is sharpest precisely when a learned controller is deployed, because the controller changes the distribution of operating points away from the historical operator, so any safety estimate calibrated on past data can become invalid at deployment. This is a feedback-induced distribution shift, and it breaks the exchangeability that distribution-free calibration normally relies on.

This paper asks: can we provide a distribution-free finite-sample guarantee that risk-controls a black-box controller's N-1 thermal exposure, cheaply, and robustly to the controller's own shift? We answer yes, with a method we call Audited Selective Verification (ASV-N1).

The key idea is to stop trusting the cheap surrogate for safety and instead use it only to save work. In each control window, ASV-N1 (i) scores every credible contingency with a cheap linear surrogate, (ii) draws a small random audit and runs full AC power flow on it, and (iii) uses the audited (score, violation) pairs to calibrate, by a Learn-Then-Test risk-control procedure \cite{ltt2021,rcps2021}, the largest surrogate threshold below which the skipped contingencies are certified to violate with probability at most a budget. Contingencies above the threshold are verified by full AC; those below it, outside the audit, are trusted. Because the certificate rests on real verification and an audit drawn from the deployment distribution, its validity does not depend on whether the surrogate is accurate, and it holds under arbitrary shift. A poor surrogate does not make the certificate unsafe; it only forces more contingencies to be verified, raising compute cost.

We first show why simpler designs do not suffice. A static threshold certificate, and even an adaptive one based on adaptive conformal inference \cite{gibbs2021}, control the violation rate only when the surrogate separates safe from unsafe contingencies in the operating regime. On three public systems this holds on IEEE 300-bus but fails on IEEE 118-bus and the 1354-bus PEGASE system, where under stress the surrogate labels as safe many contingencies that in fact violate. This negative result motivates verification-backed validity.

\textbf{Contributions.}
\begin{enumerate}
\item We introduce a distribution-free, per-window risk-controlled N-1 thermal screening certificate for arbitrary black-box grid controllers, formulated as a transductive instance of Learn-Then-Test risk control \cite{ltt2021,rcps2021}, whose validity does not depend on surrogate accuracy or on distribution stationarity, and is therefore robust to controller-induced feedback shift (\Cref{sec:method,sec:guarantee}). It controls a violation-rate budget over the screened contingencies, not worst-case N-1 security for every individual skipped outage. To our knowledge this is the first such certificate for N-1 thermal screening, although the underlying risk-control machinery is standard.
\item We give an online-audited selective-verification mechanism that combines cheap-surrogate triage with Learn-Then-Test calibration of the skip threshold, plus an adaptive audit-sizing rule that minimizes total AC-solve cost (\Cref{sec:method}).
\item We establish a negative result that motivates the design: static and adaptive threshold certificates on a cheap surrogate are valid only under a surrogate-stability condition that fails on real systems, and deterministic screening is unsafe under shift (\Cref{sec:fail,sec:experiments}).
\item We provide an open, reproducible evaluation on public systems (IEEE 118-bus and 300-bus, PEGASE 1354-bus, with the pandapower solver \cite{pandapower2018}) reporting the safety-versus-cost tradeoff under both load-induced and genuine controller-induced shift (\Cref{sec:experiments}).
\end{enumerate}

The method certifies the action of any controller, learned or not; learned controllers are our motivation, and we demonstrate robustness to a genuine controller-induced shift using an optimization-based redispatch policy, leaving an end-to-end study with a trained learning-based controller to future work. In operational terms, ASV-N1 is an online transmission security-assessment and risk-based contingency-screening layer for bulk power systems: it complements security-constrained optimal power flow and online contingency analysis and does not replace deterministic verification where policy mandates it. The strongest single result is that ASV-N1 stays within budget on all three systems while a deterministic screen is unsafe under shift, while skipping the majority of full power-flow solves where the surrogate is reliable.

\section{Related Work}
\label{sec:related}

\paragraph{Contingency screening and security assessment}
Fast N-1 screening is a mature part of operations. Linear sensitivity factors (power transfer distribution factors and line outage distribution factors) estimate post-contingency flows in closed form, and screening rules rank and skip contingencies whose estimated loading is well below limits, reserving full AC analysis for the rest \cite{wood2013,stott1987,ejebe1979}. Security-constrained optimal power flow (SCOPF) embeds N-1 constraints directly into the dispatch optimization \cite{monticelli1987,capitanescu2011}, and recent work studies adversarially robust learning of such constraints \cite{robustscopf2021}. Unlike SCOPF, which re-optimizes the dispatch to satisfy worst-case N-1 constraints, ASV-N1 does not change the dispatch: it gives an explicit finite-sample statistical bound on the violation rate among the contingencies it skips, per control window, for whatever action the controller proposes, and can therefore wrap an SCOPF, an RL policy, or any other controller as a verification-backed overlay. These methods are fast and widely used, but they are deterministic: they assume a fixed operating point and provide no probabilistic guarantee that a skipped contingency is actually safe. We show that under deployment shift this assumption fails and a deterministic screen can skip contingencies that violate. ASV-N1 keeps the same linear surrogate for efficiency but replaces the deterministic skip rule with an audited, distribution-free certificate. Beyond deterministic screening, risk-based security assessment quantifies operational risk across contingencies for control-room decisions \cite{onlinerbsa2003,riskvsdeterministic2007}, but these methods estimate expected risk under assumed contingency probabilities rather than bounding the screened violation rate. Unlike prior risk-based security assessment, ASV-N1 provides finite-sample, per-window distribution-free control of the screened violation rate using online AC-backed auditing, valid under arbitrary deployment shift.

\paragraph{Distribution-free uncertainty and risk control}
Conformal prediction and its risk-control extensions provide finite-sample, distribution-free guarantees for black-box predictors \cite{vovk2005}. Risk-controlling prediction sets and the Learn-Then-Test framework calibrate a decision threshold so that a monotone risk is controlled with high probability \cite{rcps2021,ltt2021}, and conformal risk control extends this to general monotone losses \cite{crc2024}. Standard guarantees assume exchangeability; weighted conformal prediction handles covariate shift when the likelihood ratio is known \cite{tibshirani2019}, adaptive conformal inference maintains long-run coverage under arbitrary shift by online threshold updates \cite{gibbs2021}, and analyses beyond exchangeability bound the coverage loss under distribution drift \cite{barber2023}. Crucially, these guarantees either assume exchangeability between calibration and deployment, or they remain valid only under a shift that is explicitly modeled (for example a known or estimated likelihood ratio) \cite{tibshirani2019,gibbs2021,barber2023}. Recent work brings conformal ideas closer to control and to power systems: conformal policy control calibrates how far a deployed policy may deviate from a reference \cite{conformalpolicy2026}, decision-calibrated prediction sets calibrate forecast uncertainty so a robust optimization stays feasible \cite{decisioncal2026}, and, in robotics, adaptive conformal prediction has been combined with control-barrier functions to enforce state constraints during safe reinforcement learning \cite{adaptivecpsaferl2025}. All of these either presume exchangeability or a modeled shift, or they trust the cheap predictor whose error they calibrate. ASV-N1 differs on exactly this point: it is a transductive instance of Learn-Then-Test in which the calibration set is an online audit drawn from the deployment distribution itself, so its validity is robust to arbitrary controller-induced feedback shift and does not rely on surrogate accuracy. We further show, as a negative result, that directly applying a static or adaptive threshold certificate to a cheap surrogate score is insufficient for N-1 safety when the surrogate is unreliable under stress.

\paragraph{Safe learning for power systems}
Safe reinforcement learning for power systems has grown rapidly \cite{saferlreview2024}. Shielding constrains a learning agent to a precomputed safe action set \cite{shielding2018}, and decision-calibrated prediction sets calibrate forecast uncertainty so that a downstream robust optimization remains feasible \cite{decisioncal2026}. These approaches either modify the learning process, assume an exact shield, or certify uncertainty in inputs rather than the security of a chosen action. ASV-N1 differs in three ways: it certifies the action of any frozen black-box controller after the fact, its guarantee is backed by real AC verification rather than by an assumed-exact model, and it targets per-action N-1 security under controller-induced shift. Benchmarks for learning-based grid control are also emerging \cite{rl2grid2025,safepowergraph2024,commonpower2024}; our open evaluation of the safety-versus-cost tradeoff complements them by scoring safety certificates rather than control performance.

\section{Problem Setup}
\label{sec:setup}

\subsection{N-1 thermal security of a control action}
Consider a transmission network operated in discrete control windows. A window is one control decision: the controller commits to an operating point, and that point holds until the next decision. The natural window is therefore a single operating point. As we show in \Cref{sec:experiments}, a single operating point yields a safe but conservative certificate, because the audit is small; batching the credible contingencies of several consecutive operating points, or monitoring a wider credible set, enlarges the per-window sample and tightens the certificate at the cost of a coarser time resolution. We treat a controller as a black box: it may be a reinforcement learning policy, a model predictive controller, or any other decision rule producing an operating point, that is, a setting of generation, load, and topology. Let $\mathcal{C}=\{1,\dots,N\}$ index the credible single-line N-1 contingencies for that operating point. For contingency $i$, let $V_i \in \{0,1\}$ be the true post-contingency thermal-violation label, with $V_i=1$ if and only if the full AC power flow after outage $i$ overloads some monitored line beyond its rating. Computing $V_i$ requires a full AC solve, which is the cost we wish to avoid.

A cheap surrogate assigns each contingency a danger score $r_i\in\mathbb{R}$, with larger $r_i$ meaning more dangerous. We use a standard linear screen: the post-contingency line loadings are estimated from the AC base-case flows redistributed by the line outage distribution factors, and $r_i$ is the worst estimated overload over monitored lines. The score is available for all contingencies without any contingency AC solve. We make no assumption that the surrogate is accurate.

\subsection{The certification problem}
Fix a safety budget $\alpha\in(0,1)$ and a confidence level $1-\delta$. The operator wants to act on the controller's operating point only if the contingencies it does not verify are, with high confidence, safe to skip. Concretely, we seek a decision rule that partitions $\mathcal{C}$ into a verified set (on which full AC is run) and a trusted set (skipped, no AC), such that the violation rate among the trusted set is at most $\alpha$ with probability at least $1-\delta$, while keeping the verified set, hence the compute cost, small. The challenge is that the joint distribution of $(r_i,V_i)$ at deployment is unknown and, because the controller shapes the operating point, generally differs from any historical distribution. We therefore require a guarantee that holds for an arbitrary deployment distribution and does not presume the surrogate is reliable.

\section{Why Threshold Certificates Fail}
\label{sec:fail}

A natural first design is a threshold certificate: calibrate a single skip threshold on the surrogate score using historical (score, violation) pairs, then trust every deployed contingency whose score falls below it. We summarize why this design, and even its adaptive refinement, does not give an N-1 safety guarantee under deployment shift. This negative result motivates the audited method of \Cref{sec:method}.

\paragraph{Static threshold certificate}
Calibrate $\tau$ as the largest threshold whose historical skipped-set violation rate is at most $\alpha$, then skip $\{i: r_i\le\tau\}$ at deployment. This is exactly a split-conformal (risk-controlling) certificate calibrated on a historical holdout \cite{rcps2021,crc2024}, the natural distribution-free baseline; it is valid under exchangeability between calibration and deployment. This controls the violation rate at deployment only if the score-conditional violation probability is stable between calibration and deployment. When a learned controller drives the system into a more stressed regime, the same score corresponds to a higher true violation probability, and the skipped set violates above $\alpha$. We observe exactly this: across our systems the static certificate is valid on IEEE 300-bus but its deployed skipped-set violation reaches $0.5$ to $0.86$ on IEEE 118-bus and PEGASE, far above an $\alpha=0.15$ budget.

\paragraph{Adaptive threshold certificate}
Adaptive conformal inference \cite{gibbs2021} updates the threshold online from realized outcomes to maintain a long-run rate. We implement this for the skip threshold. It improves matters but does not fix them: when the surrogate labels truly violating contingencies as safe, no threshold on the score can exclude them, so the achievable violation rate is floored by the violation rate in the surrogate's safest score bin. On PEGASE, in the stressed regime, $97$ percent of contingencies that all violate receive a score indicating no predicted overload; the adaptive certificate therefore cannot bring the violation rate below roughly $0.5$. \Cref{fig:fail} shows the cumulative skipped-set violation rate of the static and adaptive certificates as the deployment drifts: both exceed the budget on the sharp-transition systems, while IEEE 300-bus, where the surrogate separates violations, stays safe.

\begin{figure}[t]
\centering
\includegraphics[width=\linewidth]{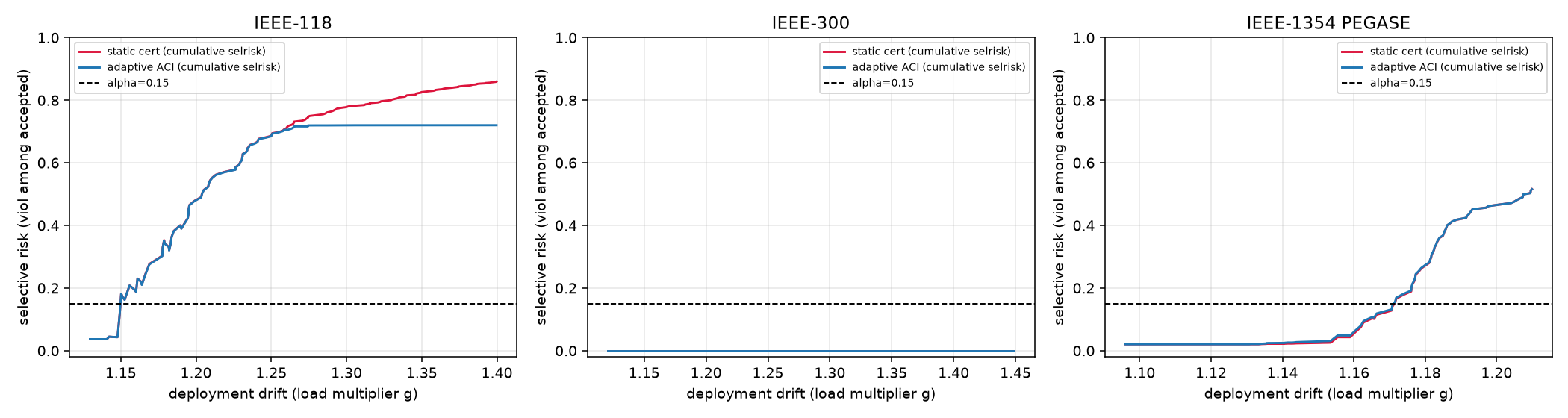}
\caption{Threshold certificates fail under shift. Cumulative skipped-set violation rate as the deployment drifts (load multiplier) for the static and adaptive (ACI) threshold certificates. Both exceed the budget $\alpha=0.15$ on IEEE 118-bus and PEGASE, where the cheap surrogate mislabels stressed contingencies as safe; only IEEE 300-bus, where the surrogate separates violations, stays safe. This motivates verification-backed validity.}
\label{fig:fail}
\end{figure}

A weighted-conformal certificate \cite{tibshirani2019} could in principle correct for the shift, but only if the deployment-to-history likelihood ratio is known or estimated, which under controller-induced shift is exactly the unknown quantity; the audit in \Cref{sec:method} sidesteps this by calibrating on a fresh on-distribution sample instead of modeling the shift. The lesson is that any certificate whose safety rests on the cheap surrogate inherits the surrogate's blind spots. ASV-N1 instead makes safety rest on real verification, using the surrogate only to decide what to verify.

\section{Audited Selective Verification}
\label{sec:method}

ASV-N1 wraps a controller and certifies each control window independently. The mechanism has three parts: a cheap triage, an online audit, and a calibrated skip threshold. \Cref{alg:asv} states it.

\subsection{Triage, audit, and calibration}
Given the window's contingency set $\mathcal{C}$ with scores $\{r_i\}$, ASV-N1 draws an audit $A\subseteq\mathcal{C}$ that includes each index independently with probability $\pi$ (a fixed-size uniform sample behaves identically). Crucially, $A$ is drawn from indices only, independently of the scores and labels. Full AC power flow is run on $A$, revealing $\{V_i: i\in A\}$.

For a candidate threshold $\tau$, let the skip set be $\Skip(\tau)=\{i\in\mathcal{C}: r_i\le\tau\}$, a set fixed by the scores. Let $U_\delta(k,n)$ denote the Clopper-Pearson upper $(1-\delta)$ confidence limit for a binomial rate from $k$ violations in $n$ trials. Order the distinct candidate thresholds $\tau^{(1)}<\dots<\tau^{(M)}$ (a data-independent order, strictest skip set first). Using only the audited contingencies, ASV-N1 tests these thresholds in sequence and accepts $\tau^{(j)}$ when $U_\delta(\sum_{i\in A, r_i\le\tau^{(j)}} V_i,\ |A\cap\Skip(\tau^{(j)})|)\le\alpha$. It stops at the first threshold that is not accepted and sets $\tau^\star$ to the largest threshold in the accepted prefix:
\begin{equation}
\tau^\star = \tau^{(j^\star)},\quad j^\star=\max\{\, j: \text{$\tau^{(1)},\dots,\tau^{(j)}$ all accepted}\,\}.
\label{eq:tau}
\end{equation}
This fixed-sequence (prefix) rule, rather than a global maximum over accepted thresholds, is what controls multiplicity and makes the guarantee of \Cref{sec:guarantee} valid even when the audited confidence limit is not monotone in $\tau$; it may be conservative if the accepted set has gaps. ASV-N1 then verifies (full AC) every contingency with $r_i>\tau^\star$; among $r_i\le\tau^\star$, the audited ones are already solved, and the remaining trusted contingencies are skipped.

\begin{algorithm}[t]
\caption{ASV-N1 (one control window)}
\label{alg:asv}
\textbf{Input:} contingencies $\mathcal{C}$, scores $\{r_i\}$, budget $\alpha$, confidence $1-\delta$.
\begin{enumerate}
\item Draw audit $A\subseteq\mathcal{C}$ (label-independent random sample); run AC on $A$ to obtain $\{V_i:i\in A\}$.
\item Select $\tau^\star$ by \eqref{eq:tau} via fixed-sequence testing over the ordered thresholds.
\item Verify (full AC) all $i$ with $r_i>\tau^\star$.
\item Trust (skip) all $i\notin A$ with $r_i\le\tau^\star$.
\end{enumerate}
\textbf{Output:} verified contingencies are AC-checked exactly; the trusted skipped contingencies are risk-controlled at the certified budget (skip-set rate at most $\alpha$, hence trusted-set rate at most $\alpha/(1-f)$, with confidence $1-\delta$).
\end{algorithm}

\subsection{Adaptive audit sizing}
The per-window cost is the audit size plus the number of verified contingencies. A larger audit tightens the Clopper-Pearson bound in \eqref{eq:tau}, which can raise $\tau^\star$ and shrink the verified set, so there is a cost-optimal audit size. Two variants preserve the guarantee. (i) Score-only sizing: choose the audit size as a function of the score distribution alone, for example to place a target number of audited contingencies in the candidate skip region. Because this uses only scores, not labels, \Cref{asm:audit} holds and \Cref{thm:main} applies directly. (ii) Cost-optimal sizing: grow a fixed grid of $G$ nested audit prefixes and, for each, predict the total AC count (audit plus verify), then select the size minimizing the predicted total. Because this selection peeks at audited labels, validity is restored by a union bound over the grid, replacing $\delta$ by $\delta/G$ in \eqref{eq:tau}; since $G$ is small (a handful of sizes), the resulting widening of the confidence interval is mild. All experiments in this paper use variant (ii), cost-optimal sizing with the $\delta/G$ union-bound correction and $G=8$ candidate sizes (\Cref{prop:adaptive}); variant (i) is stated for completeness.

\subsection{Robustness to controller-induced shift}
Because the audit is drawn from the current window, the calibration in \eqref{eq:tau} uses deployment-distribution data. Re-running it every window means the certificate never extrapolates from a stale historical distribution. This is precisely what makes the guarantee robust to the shift the controller induces: the audit observes the controller's own operating points. The surrogate enters only through $\tau^\star$ and therefore affects cost, not validity; a useless surrogate yields $\tau^\star$ that verifies everything, which is expensive but still safe.

\section{Guarantee}
\label{sec:guarantee}

We state the per-window safety guarantee. The setting is transductive and distribution-free: we treat each control window as fixed and randomize only the audit, so the guarantee is per window and exact, not asymptotic over time. Here the risk is simply the fraction of skipped contingencies that violate their thermal limits.

\begin{assumption}[Label-independent audit]
\label{asm:audit}
The audit is a uniform random sample of the contingency set drawn independently of the violation labels $\{V_i\}$. Its size may be fixed in advance or chosen from label-independent information (including the scores $\{r_i\}$, as in \Cref{prop:adaptive}); conditional on the selected size, the audited contingencies form a uniform, label-independent sample.
\end{assumption}

\begin{assumption}[Trusted oracle and thermal scope]
\label{asm:oracle}
The AC N-1 power flow is the trusted verifier defining $V_i$, and $V_i$ encodes thermal-limit violations under the AC model. Voltage and reactive limits are outside the present statement and are reported separately as a diagnostic.
\end{assumption}

Order the distinct candidate thresholds $\tau^{(1)}<\dots<\tau^{(M)}$ (this order depends on the scores, which are fixed, not on the labels). Let $\Rrisk(\tau)=|\Skip(\tau)|^{-1}\sum_{i\in\Skip(\tau)}V_i$ be the true violation rate of the whole skip set.

\begin{lemma}[Valid audited p-value]
\label{lem:pvalue}
Fix $\tau$ and condition on the audited skip count $n=|A\cap\Skip(\tau)|$. Under the null $H:\Rrisk(\tau)>\alpha$, the audited violation count $K$ is the number of violations in a label-independent size-$n$ sample drawn without replacement from $\Skip(\tau)$, hence hypergeometric with mean $n\,\Rrisk(\tau)>n\alpha$. Consider the left-tail statistic $P=\Pr[\mathrm{Bin}(n,\alpha)\le K]$ (equivalently, $U_\delta(K,n)\le\alpha$). The hypergeometric lower tail is dominated by the binomial lower tail at the same mean \cite{hoeffding1963,serfling1974}, and since $\Rrisk(\tau)>\alpha$ places more mass on large $K$ than $\mathrm{Bin}(n,\alpha)$ does, small values of $P$ occur with at most their nominal probability: $\Pr[P\le\delta\mid n]\le\delta$. Marginalizing over $n$ preserves this, so $P$ is a valid (super-uniform) p-value, and the Clopper-Pearson rule is, if anything, conservative for the without-replacement audit.
\end{lemma}

\begin{theorem}[Per-window distribution-free safety]
\label{thm:main}
Under \Cref{asm:audit,asm:oracle}, with probability at least $1-\delta$ over the audit, the threshold $\tau^\star$ selected by \eqref{eq:tau} with fixed-sequence testing satisfies
\begin{equation}
\Rrisk(\tau^\star)\ \le\ \alpha .
\end{equation}
This holds for any fixed $\{(r_i,V_i)\}$ (arbitrary deployment distribution, including controller-induced shift) and any surrogate, with no assumption on surrogate accuracy or cross-window stationarity.
\end{theorem}

\begin{proof}[Proof sketch]
The candidate-threshold order is data-independent. Fixed-sequence testing that stops at the first non-rejection, using the per-hypothesis level-$\delta$ valid p-values of \Cref{lem:pvalue}, controls the family-wise error rate at $\delta$ under arbitrary dependence \cite{holm1979,ltt2021}. Hence with probability at least $1-\delta$ every rejected hypothesis is truly false, so $\Rrisk(\tau^{(j)})\le\alpha$ for each selected $j$, and in particular $\Rrisk(\tau^\star)\le\alpha$. This is the transductive Learn-Then-Test instance with the audit as calibration and the fixed skip population as the risk object \cite{ltt2021,rcps2021}. A full proof is in \Cref{app:proof}.
\end{proof}

\begin{corollary}[Trusted-set bound, deterministic transfer]
\label{cor:subset}
On the event $\Rrisk(\tau^\star)\le\alpha$ of \Cref{thm:main}, the trusted set $T=\{i\notin A: r_i\le\tau^\star\}$ satisfies, deterministically,
\begin{equation}
\begin{split}
\frac{1}{|T|}\sum_{i\in T}V_i\ &\le\ \frac{\Rrisk(\tau^\star)\,|\Skip(\tau^\star)|}{|T|}\ \le\ \frac{\alpha}{1-f},\\
&f=\frac{|A\cap\Skip(\tau^\star)|}{|\Skip(\tau^\star)|},
\end{split}
\end{equation}
where $f$ is the audited fraction of the skip set (assuming $f<1$, that is, a nonempty trusted set). This holds with probability at least $1-\delta$ (the same event). It requires no assumption that $T$ is a clean random subsample after the data-dependent selection of $\tau^\star$: the violation count in $T$ is at most the violation count in the whole set $\Skip(\tau^\star)$, which $\Rrisk(\tau^\star)\le\alpha$ bounds. We therefore treat the skip-set rate $\Rrisk(\tau^\star)$ as the primary certified quantity; the trusted (unverified) rate is at most $\alpha/(1-f)$, and with the audit fractions used here ($f$ near $0.15$ to $0.2$) this inflation is small. Audited contingencies, being verified, carry no residual risk.
\end{corollary}

\begin{proposition}[Cost-optimal audit sizing]
\label{prop:adaptive}
Fix a data-independent grid of $G$ candidate audit sizes (nested prefixes of a label-independent sampling order). If the audit size is chosen from this grid by any rule that may depend on the audited labels, then running the selection of \eqref{eq:tau} at level $\delta/G$ for each candidate and reporting the chosen one preserves \Cref{thm:main} at level $1-\delta$, by a union bound over the $G$ candidates. Our experiments use $G=8$.
\end{proposition}

\begin{corollary}[Long-run budget]
\label{cor:longrun}
Over $H$ windows, applying \Cref{thm:main} at confidence $1-\delta/H$ guarantees every window's trusted set is safe at level $\alpha$ simultaneously with probability at least $1-\delta$. A time-uniform confidence sequence \cite{howard2021} avoids the $1/H$ factor and handles adaptively chosen windows. No cross-window stationarity is assumed.
\end{corollary}

\paragraph{Skip set versus trusted set}
\Cref{thm:main} controls the violation rate of the entire skip set $\Skip(\tau^\star)$, which includes the audited contingencies that ASV-N1 has already verified by full AC. The quantity an operator acts on, and the quantity we report in \Cref{sec:experiments}, is the violation rate of the trusted set $T$ of skipped contingencies that are \emph{not} verified. \Cref{cor:subset} transfers the skip-set guarantee to exactly this trusted set, so the reported trusted-set rate is the certified quantity; the audited contingencies carry no residual risk because their true labels are known.

\paragraph{Operational reading}
In practice, \Cref{thm:main} gives the operator a tunable safety budget with a confidence level. Setting $\alpha=0.05$ and $1-\delta=0.95$ means: in each control window, among the full set of contingencies the tool deems skippable, at most five percent are at risk of a post-contingency thermal violation; among the unverified subset it actually trusts without a check, the bound is the slightly looser $\alpha/(1-f)$, about six percent at the audit fractions used here; and these statements are correct at least ninety-five percent of the time. The operator never trusts the cheap surrogate blindly: every contingency the surrogate cannot vouch for, at the certified threshold, is verified by a full power-flow study, and a small random sample of the rest is verified as an ongoing check. If the surrogate degrades (for example because a new controller drives the system into an unfamiliar regime), the audit detects it and the tool automatically verifies more contingencies, trading compute for the same safety budget rather than silently becoming unsafe.

\begin{proposition}[Cost and the role of the surrogate]
\label{prop:cost}
The per-window AC-solve count is $|A\cup\{i: r_i>\tau^\star\}|$, that is $|A|+|\{i\notin A: r_i>\tau^\star\}|$ (the audited contingencies above the threshold are counted once, not twice); the AC-solve fraction is this divided by $N$. Validity is independent of the surrogate; the surrogate enters only through $\tau^\star$. A surrogate that ranks violations well admits a larger safe $\tau^\star$ and lower cost; a useless surrogate yields verify-all, which is maximal cost but still valid.
\end{proposition}

\section{Experiments}
\label{sec:experiments}

\subsection{Setup}
We evaluate on three public transmission systems using the open-source pandapower solver \cite{pandapower2018}: IEEE 118-bus, IEEE 300-bus, and the 1354-bus PEGASE system. For each system we equip lines with thermal ratings so that the network is N-1 secure at nominal load with a headroom margin, then create operating points by scaling load with heterogeneous noise; generation follows load only partially, producing congestion. The aim is not to reproduce a specific utility operating condition but to test certificate validity and cost under controlled, reproducible stress regimes on public networks. For every operating point and every credible single-line contingency we compute the cheap surrogate score $r$ (line outage distribution factors applied to the AC base-case flows) and the true label $V$ (a full AC contingency solve). The label encodes thermal violations only: a contingency is labeled violating if its converged AC solution overloads a monitored line. The credible N-1 set excludes single-line outages that island the network, which are non-credible for thermal screening and appear as non-convergent solutions (islanding is detected as failure of the post-contingency AC power flow to converge, that is, the outage disconnects part of the network and is treated as a load-loss event rather than a thermal-screening case); on IEEE 300-bus the per-contingency islanding (non-convergence) rate is about $5$ percent. We treat such cases as outside the thermal scope of \Cref{asm:oracle}; a sensitivity test that instead counts every non-convergent case as a violation leaves the certificate within budget (trusted-set violation $0.000$ on IEEE 300-bus). For IEEE 118-bus we monitor all $173$ single-line contingencies, a complete N-1 sweep; for IEEE 300-bus we monitor the full credible set of about $268$ of $283$ single-line outages per operating point (about $15$ island) across $100$ operating points whose base case converges; for the 1354-bus PEGASE system we monitor all $1751$ single-line contingencies; all but about two per operating point are credible (almost none island), so this too is a near-complete N-1 sweep, at scale. In operations the credible contingency set is defined by the utility; ASV-N1 certifies the violation rate restricted to whatever set is monitored. The certified result on IEEE 300-bus is the same under the full set and a smaller random sample, confirming the evaluation is not a sampling artifact. Unless stated otherwise the budget is $\alpha=0.15$ and confidence $1-\delta=0.9$. The formal certificate (\Cref{thm:main}) controls the skip-set violation rate at $\alpha$; the unverified trusted subset is bounded by $\alpha/(1-f)$ (\Cref{cor:subset}). We report the realized (empirical) violation rate among trusted (skipped, unverified) contingencies, which is well within budget, and the AC-solve fraction relative to a full N-1 sweep, where lower is cheaper. Implementation details: the contingency sets are as above; the audit is a label-independent uniform sample whose size is chosen by the adaptive rule of \Cref{sec:method}; the budget is swept over $\alpha\in\{0.05,0.10,0.15,0.20\}$ and the confidence is $1-\delta=0.9$; all random seeds are fixed (data generation and audit). All systems, the pandapower solver \cite{pandapower2018}, and the contingency definitions are public; an anonymized code package that reproduces every table and figure accompanies this submission for review, and a permanent public repository will host it on publication. Experiments ran on a workstation with a 20-core CPU, 121 GB of RAM, and an NVIDIA GB10 GPU; the pandapower AC power flows are CPU Newton-Raphson solves (the GPU was not used for the power-flow oracle).

\subsection{Validity and cost across systems}
\Cref{tab:main} reports ASV-N1 with adaptive audit sizing at two window granularities. On all three systems the realized trusted-set violation rate is well within the budget, consistent with \Cref{thm:main}, in both settings. At the single operating point, which is the real-time control window, ASV-N1 saves $29$ to $75$ percent of the full N-1 solves (AC-solve fraction $0.25$ to $0.71$). The savings come from the surrogate: where it reliably separates safe from unsafe contingencies, the audit certifies a large skip set; where it does not, the certificate verifies more, which is why IEEE 300-bus (a well-separated case) saves most. Batching the contingencies of several operating points into one window tightens the Clopper-Pearson bound and raises the savings to $36$ to $80$ percent, at a coarser time resolution; this is a statistical-efficiency option, not a property of a single control decision. The overall violation rates span $4$ percent (IEEE 300-bus) to $63$ percent (IEEE 118-bus), so the evaluated operating points range from lightly loaded (where the surrogate separates well and savings are large) to heavily congested, and the savings are not an artifact of a single regime. In wall-clock terms, a single AC N-1 contingency solve takes about $12$ ms on IEEE 300-bus while the surrogate score is a matrix operation costing well under a millisecond, so runtime tracks the AC-solve fraction; the single-operating-point sweep of about $268$ credible contingencies takes roughly $3$ s, which ASV-N1 reduces to about $0.8$ s, and the audit contingencies are independent and parallelizable.

\paragraph{Minimum audit size}
The audit must be large enough for the Clopper-Pearson bound to license any skips. With zero audited violations, certifying a skipped-set rate at $\alpha$ with confidence $1-\delta$ requires at least $n_{\min}=\lceil \ln\delta/\ln(1-\alpha)\rceil$ audited skipped contingencies: $n_{\min}=15$ at $\alpha=0.15$, $22$ at $\alpha=0.10$, and $45$ at $\alpha=0.05$ (all at $\delta=0.1$). For the adaptive rule, which selects the audit size from $G=8$ candidates and therefore tests each at $\delta/G=0.0125$ (\Cref{prop:adaptive}), $n_{\min}$ rises to $27$ at $\alpha=0.15$. We floor the single-operating-point audit at $20$ percent of the window, which exceeds $n_{\min}$ on every system in \Cref{tab:main}. If a window is too small to reach $n_{\min}$, no threshold is certifiable and ASV-N1 verifies the whole window, safe by construction but with no savings.

\Cref{tab:alpha} sweeps the budget at the single operating point. Cost rises monotonically as $\alpha$ tightens: the operator buys stricter safety with more AC solves. An AC-solve fraction of $1.00$ means the window is smaller than the audit needed to certify that budget, so the certificate conservatively verifies everything (safe by construction, no savings). Strict budgets such as $\alpha=0.05$ remain feasible on IEEE 300-bus and PEGASE, and $\alpha=0.01$ on PEGASE; on the smaller systems they are recovered by batching. Throughout, the realized trusted-set violation stays within budget.

\begin{table}[t]
\centering
\caption{Budget sweep at the single operating point: AC-solve fraction as the risk budget $\alpha$ tightens ($\delta=0.1$). Cost rises monotonically as $\alpha$ falls; $1.00$ means the window is too small to certify that $\alpha$, so ASV-N1 verifies all (safe, no savings). The realized trusted-set violation is within budget in every cell.}
\label{tab:alpha}
\begin{tabular}{lccccc}
\toprule
System & $\alpha{=}0.20$ & $0.15$ & $0.10$ & $0.05$ & $0.01$ \\
\midrule
IEEE 118-bus & 0.71 & 0.71 & 0.75 & 1.00 & 1.00 \\
IEEE 300-bus & 0.24 & 0.25 & 0.26 & 0.37 & 1.00 \\
PEGASE 1354 & 0.47 & 0.47 & 0.47 & 0.53 & 0.76 \\
\bottomrule
\end{tabular}
\end{table}

\begin{table}[t]
\centering
\caption{ASV-N1 with adaptive audit sizing ($\alpha=0.15$) at two window granularities. ``Single operating point'' is the real-time setting (one control decision; about $173$, $268$, and $1749$ credible contingencies for IEEE 118-bus, IEEE 300-bus, and PEGASE, respectively). ``Batched'' windows contain about $1000$ to $1500$ contingencies depending on system and partitioning: on IEEE 118-bus and 300-bus this pools several consecutive operating points, whereas on PEGASE a single operating point already holds a large credible set, so batching there mainly serves the many-window statistical test rather than increasing the per-window count. Batching tightens the Clopper-Pearson bound and lowers cost at a coarser time resolution. The trusted-set violation is the realized (empirical) rate among unverified skipped contingencies; it is well within budget in both settings, while the formal certificate controls the skip-set rate at $\alpha$ (\Cref{thm:main}; trusted subset at most $\alpha/(1-f)$). AC-solve fraction is relative to a full sweep (lower is cheaper).}
\label{tab:main}
\setlength{\tabcolsep}{3pt}
\begin{tabular}{lccccc}
\toprule
& & \multicolumn{2}{c}{Single op. point} & \multicolumn{2}{c}{Batched} \\
\cmidrule(lr){3-4}\cmidrule(lr){5-6}
System & Overall viol. & trusted & AC-frac & trusted & AC-frac \\
\midrule
IEEE 118-bus & 0.63 & 0.012 & 0.71 & 0.020 & 0.64 \\
IEEE 300-bus & 0.04 & 0.005 & 0.25 & 0.006 & 0.20 \\
PEGASE 1354 & 0.34 & 0.011 & 0.47 & 0.010 & 0.37 \\
\bottomrule
\end{tabular}
\end{table}

\begin{figure*}[!t]
\centering
\includegraphics[width=\textwidth]{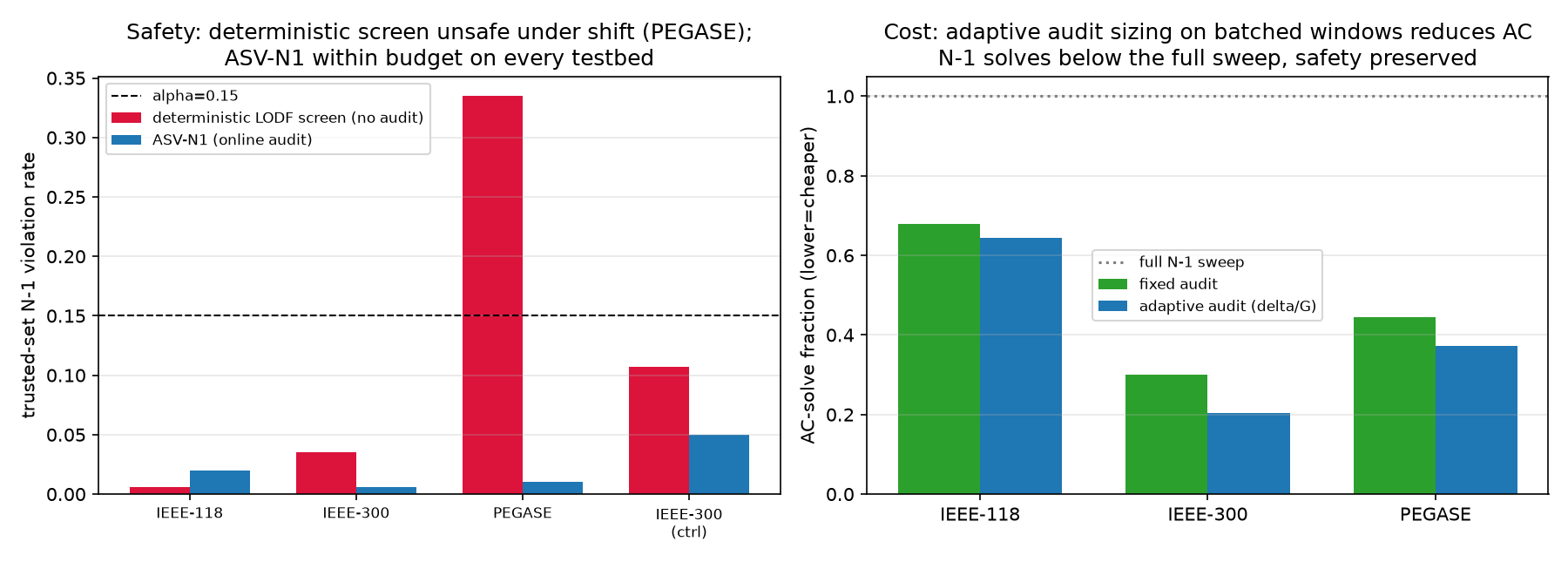}
\caption{Left: safety. A deterministic surrogate screen (no audit) is unsafe under shift on PEGASE (skipped-set N-1 violation $0.34$ against the $0.15$ budget), while ASV-N1 is safe on all four testbeds including a genuine controller-induced shift. Right: cost. Adaptive audit sizing lowers the AC-solve fraction on every system (IEEE 300-bus reaches $0.20$, an $80$ percent reduction), all below the full N-1 sweep, with validity preserved.}
\label{fig:final}
\end{figure*}

\subsection{Testing the probabilistic guarantee across many windows}
A single trusted-set violation number does not test an $(\alpha,\delta)$ guarantee, which is a statement about the distribution of outcomes across windows: \Cref{thm:main} requires that the per-window violation rate exceed $\alpha$ in at most a $\delta$ fraction of windows. We test this directly. For each system we form many control windows by randomly partitioning the labeled pool into windows of about one thousand contingencies, repeat over forty independent seeds (which also re-randomize the audit), and record the realized trusted-set violation rate of each window. \Cref{fig:multiwindow} plots the resulting cumulative distributions. With $\alpha=0.15$ and $1-\delta=0.9$, the fraction of windows that breach $\alpha$ is $0.058$ on IEEE 118-bus and $0.000$ on IEEE 300-bus and PEGASE, all at or below $\delta=0.1$, across $1360$, $1000$, and $4160$ windows respectively. The empirical breach frequencies are therefore consistent with the stated guarantee, not only in expectation but at the required confidence level, and the breach fraction approaches but does not exceed $\delta$ on IEEE 118-bus, indicating the certificate is tight rather than loose at this window size.

Window size trades tightness for conservativeness. At the operational window size of a single operating point, with about $173$, $268$, and $1749$ credible contingencies for IEEE 118-bus, IEEE 300-bus, and PEGASE respectively, the Clopper-Pearson interval is wider, so the certified threshold is stricter, the certificate verifies more and skips less, and the trusted-set violation is essentially zero: over $83040$, $61080$, and $251760$ windows for the three systems, the breach fraction is $0.000$ everywhere. The certificate is thus even safer, but cheaper savings require aggregating a larger contingency set per window (for example batching several operating points or a wider credible set), which tightens the bound as in \Cref{fig:multiwindow}.

\begin{figure}[t]
\centering
\includegraphics[width=\linewidth]{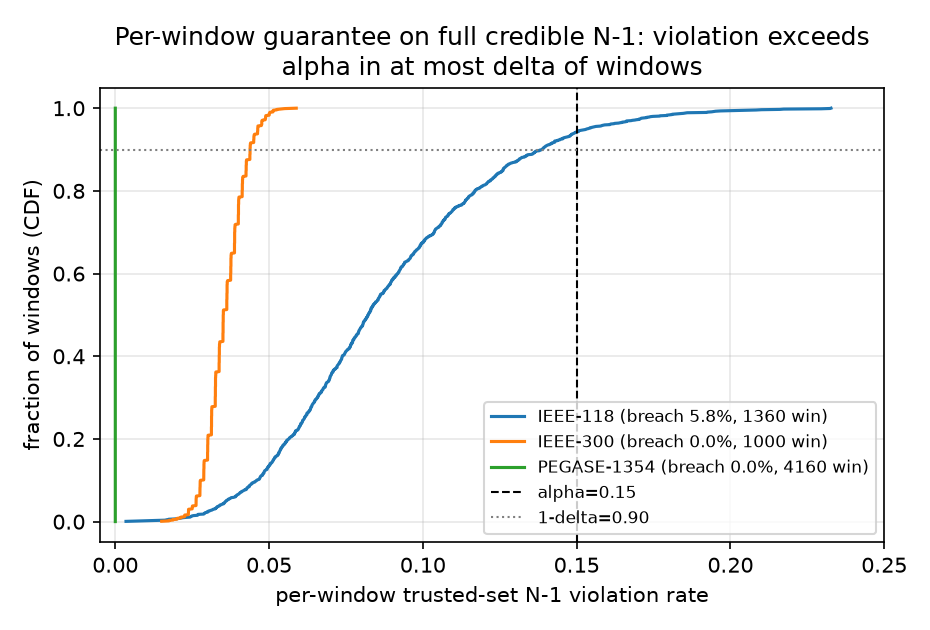}
\caption{The $(\alpha,\delta)$ guarantee holds across many windows. Cumulative distribution of the per-window trusted-set violation rate over hundreds of windows and forty seeds per system. The horizontal line at $1-\delta=0.9$ meets each curve at or before $\alpha=0.15$, so at most a $\delta$ fraction of windows breach the budget (breach fractions: IEEE 118-bus $5.8$ percent, IEEE 300-bus and PEGASE $0$ percent).}

\label{fig:multiwindow}
\end{figure}

\subsection{Comparison to deterministic screening}
We compare ASV-N1 against the operational screens practitioners use, on PEGASE under controller-induced shift (the hardest case), at budget $\alpha=0.15$ (\Cref{tab:baseline}). A deterministic screen that skips every contingency the surrogate predicts safe is cheap but unsafe (violation $0.335$), because the surrogate's confidently-safe predictions are wrong under stress; a safety margin does not rescue it ($0.336$ at a $10$ percent buffer), and a threshold calibrated on historical data breaks catastrophically when the controller shifts the operating point ($0.671$). Verifying the top-$k$ most severe contingencies at a budget matched to ASV-N1 happens to be safe here ($0.068$), but it carries no certificate and cannot know the required budget without an audit; under a different shift it degrades to the calibrated-threshold failure. ASV-N1 certifies safety ($0.011$) at the same cost as top-$k$. This is the central distinction from fast screening: ASV-N1 provides a guarantee that survives an unreliable surrogate.

\begin{table}[t]
\centering
\caption{Operational screens on PEGASE under controller-induced shift (real-time single operating point, $\alpha=0.15$). Deterministic, margin, and historically-calibrated screens are unsafe; top-$k$ at a matched budget is safe only by chance and gives no certificate. ASV-N1 certifies safety at the same cost.}
\label{tab:baseline}
\begin{tabular}{lccc}
\toprule
Screen & Trusted viol. & AC-frac & Safe? \\
\midrule
Deterministic LODF screen & 0.335 & 0.01 & no \\
Margin screen ($10\%$ buffer) & 0.336 & 0.01 & no \\
Static calibrated threshold & 0.671 & 0.00 & no \\
Top-$k$ severe (matched budget) & 0.068 & 0.47 & no cert. \\
ASV-N1 (audited) & \textbf{0.011} & 0.47 & \textbf{certified} \\
\bottomrule
\end{tabular}
\end{table}

\subsection{Adaptive audit sizing}
\Cref{fig:final} (right) compares a fixed audit fraction to the adaptive rule. The adaptive rule lowers the AC-solve fraction on every system while preserving validity; on IEEE 300-bus it falls from $0.30$ to $0.20$. All bars remain below the full N-1 sweep baseline, so ASV-N1 always saves work relative to exhaustive verification.

\paragraph{Parameter selection}
The budget $\alpha$ is the operator's chosen safety level and sets the safety-versus-cost tradeoff: a smaller $\alpha$ certifies a stricter threshold, which verifies more contingencies and raises the AC-solve fraction, while validity holds at every $\alpha$. Sweeping $\alpha\in\{0.05,0.10,0.15,0.20\}$ produces a monotone cost curve that stays below the full-sweep baseline throughout, so an operator can read off the verification cost of a desired safety level directly. The confidence $1-\delta$ enters only through the Clopper-Pearson width and has a mild logarithmic effect on cost; we use $1-\delta=0.9$. The audit size is not a free parameter: the adaptive rule selects it to minimize total cost, so the operator sets only $\alpha$ and $\delta$.

\subsection{Genuine controller-induced shift}
The previous experiments induce shift by load scaling. To test a genuine controller-induced shift, we contrast two control policies on IEEE 300-bus. Because the economic controller only becomes security-naive under stress, where the DC optimal dispatch and the AC security limits diverge, this experiment uses a more heavily loaded regime and a representative contingency sample than the scalability study; the certificate is valid over whatever contingency set is monitored. The historical operator uses proportional dispatch, while the deployed controller uses an economic DC optimal power flow redispatch that minimizes generation cost subject to the DC power balance and base-case line limits only, with no N-1 contingency or thermal-security constraints; this is what makes it security-naive under stress and a realistic source of controller-induced shift. The two policies, on the same loads, produce different operating points; the economic, security-naive controller raises the N-1 violation rate from $0.30$ to $0.61$ and the mean surrogate score from $3.75$ to $8.71$, a strong shift. \Cref{tab:ctrl} reports deploying an operator-calibrated certificate on the controller. The static (no-audit) certificate degrades to a trusted-set violation of $0.166$, above the budget, because it was calibrated on the operator distribution. ASV-N1, whose audit observes the controller's own operating points, holds the trusted-set violation at $0.043$, at a higher verification cost (AC-solve fraction $0.84$). This is the honest and expected behavior: under a strong, genuine controller-induced shift the savings shrink to about $16$ percent, because the audit correctly responds to the less reliable surrogate by verifying more. Safety is maintained throughout; the surrogate's degradation is paid in compute, not in violations.

\begin{table}[t]
\centering
\caption{Genuine controller-induced shift on IEEE 300-bus (operator proportional dispatch vs. economic DC-OPF controller; $\alpha=0.15$). The operator-calibrated static certificate degrades past the budget; ASV-N1 stays safe.}
\label{tab:ctrl}
\setlength{\tabcolsep}{4pt}
\begin{tabular}{lcc}
\toprule
Certificate & Trusted-set viol. & AC-solve frac. \\
\midrule
Static (operator-calibrated, no audit) & 0.166 & 0.54 \\
ASV-N1 (online audit on controller) & 0.043 & 0.84 \\
\bottomrule
\end{tabular}
\end{table}

\section{Discussion and Limitations}
\label{sec:discussion}

\paragraph{What the guarantee is and is not}
\Cref{thm:main} is a per-window, high-probability statement, not a per-action worst-case one: it controls the violation \emph{rate} of the skip set at level $\alpha$ with confidence $1-\delta$ (\Cref{thm:main}), so the unverified trusted subset has violation rate at most $\alpha/(1-f)$ (\Cref{cor:subset}), aggregated over windows by \Cref{cor:longrun}. This matches how operators reason about risk budgets and is what makes a distribution-free statement possible without restricting the deployment distribution. It is therefore a careful claim: the certificate controls a violation-rate budget under arbitrary deployment shift, but it does not guarantee that every individual skipped contingency is safe, nor does it provide worst-case N-1 security for each skipped outage. An operator who needs a hard per-contingency guarantee can drive the budget toward zero, at which point ASV-N1 verifies essentially everything. ASV-N1 is thus intended for risk-budgeted screening and triage, not for settings where policy mandates deterministic verification of every credible contingency; a small set of high-impact outages can be excluded from the skip set and always verified, so the certificate complements rather than replaces a deterministic check where one is required.

\paragraph{Scope}
The certified object is N-1 thermal security under the AC model (\Cref{asm:oracle}). Voltage and reactive-power limits are outside the present theorem; in our experiments we track voltage and non-convergence events separately. Extending the audit to a joint thermal-and-voltage label is direct in principle, since the audit already runs full AC, and is a natural next step. The trusted verifier is the AC N-1 solver, which is standard in operations.

\paragraph{Cost floor and surrogate quality}
The audit is an unavoidable cost: even with a perfect surrogate, ASV-N1 pays the audit fraction. \Cref{prop:cost} makes the tradeoff explicit, and \Cref{fig:final} shows the audit fraction is a small part of the total. The audit contingencies are independent and embarrassingly parallel, so the per-window wall-clock cost scales with available compute rather than with system size; our largest case, the 1354-bus PEGASE system, runs with the same procedure as the smaller ones, indicating the approach scales to realistic networks. A better surrogate raises $\tau^\star$ and lowers cost; a learned surrogate that resolves the stressed regime where linear screening fails is a promising way to push the savings on the harder systems toward those on IEEE 300-bus, without affecting validity.

\paragraph{Integration in an energy management system}
ASV-N1 is designed to sit between a controller and the dispatch it commits. Concretely, in an energy management system the flow is: (i) a controller (reinforcement learning, model-predictive control, or an optimization such as security-constrained optimal power flow) proposes an operating point; (ii) ASV-N1 computes the linear surrogate scores for the operator-defined credible contingency set, draws the audit and runs full AC power flow on the audit plus the above-threshold contingencies, and calibrates $\tau^\star$; (iii) it returns the trusted set, the verified violations, and the certified risk. If the realized or certified violation rate exceeds the operator's budget, the action is rejected or sent back for re-dispatch. The certificate is an overlay, or watchdog layer, that requires no change to the controller and reuses the AC contingency solver an operator already trusts, so it complements rather than replaces security-constrained optimal power flow.

\paragraph{Controllers}
We demonstrated a genuine controller-induced shift with an economic redispatch controller. Wrapping a trained learning-based controller end to end is future work; the method is agnostic to the controller because it certifies the resulting operating point, not the policy.

\section{Conclusion}
\label{sec:conclusion}

We presented Audited Selective Verification, a distribution-free certificate that risk-controls the N-1 thermal violation rate of a black-box controller's action, rather than asserting hard per-contingency N-1 security. By using a cheap surrogate only to decide what to verify, and resting validity on an online audit plus full AC verification, ASV-N1 holds the violation rate of the trusted contingencies at a budget with high confidence, for an arbitrary deployment distribution and any surrogate, including under the shift the controller itself induces. We showed that simpler threshold certificates, static or adaptive, fail when the surrogate is unreliable under stress, and that deterministic screening is unsafe under shift, whereas ASV-N1 keeps the realized violation rate within budget across the tested windows on IEEE 118-bus, IEEE 300-bus, and the 1354-bus PEGASE system. At the single operating point, the real-time control window, it cuts AC N-1 solves by $29$ to $75$ percent, rising to $36$ to $80$ percent when contingencies are batched across operating points, and it degrades gracefully toward verify-all exactly where the surrogate is least trustworthy. The method gives operators a tunable safety budget with a quantified, surrogate-adaptive compute cost, and a verification-backed safety layer behind which a controller can be deployed. As utilities begin to adopt learning-based controllers in operations, ASV-N1 offers a concrete watchdog layer that admits such controllers under an explicit, auditable safety budget rather than on trust, which aligns with operational practice around documented contingency assessment: the certified risk budget and the audit log give operators an auditable trail. Future work includes a joint thermal-and-voltage audit, a learned surrogate to lower cost on stressed systems, and end-to-end evaluation with trained learning-based controllers.


\appendices
\section{Full Proof of Theorem 1}
\label{app:proof}

We restate the setting. Condition on the fixed window: the pairs $\{(r_i,V_i)\}_{i=1}^N$ are arbitrary fixed quantities. The audit $A$ includes each index independently with probability $\pi$ (or is a fixed-size uniform sample), drawn independently of $\{(r_i,V_i)\}$ by \Cref{asm:audit}. For a threshold $\tau$, $\Skip(\tau)=\{i: r_i\le\tau\}$ is fixed, and $\Rrisk(\tau)=|\Skip(\tau)|^{-1}\sum_{i\in\Skip(\tau)}V_i$. The candidate thresholds $\tau^{(1)}<\dots<\tau^{(M)}$ are the distinct score values, ordered independently of the labels.

\subsection{Proof of Lemma 1}
Fix $\tau$ and condition on $n=|A\cap\Skip(\tau)|$, a function of $A$ and the scores only. By \Cref{asm:audit}, given $n$, the audited subset of $\Skip(\tau)$ is a uniform size-$n$ sample without replacement from the finite population $\Skip(\tau)$, whose violation indicators are fixed with mean $\Rrisk(\tau)$. Hence the audited violation count $K$ is hypergeometric with mean $n\,\Rrisk(\tau)$. Under $H:\Rrisk(\tau)>\alpha$, $\mathbb{E}[K\mid n]>n\alpha$.

Define $P=\Pr[\mathrm{Bin}(n,\alpha)\le K]$, the binomial left-tail evaluated at the observed $K$; this equals the Clopper-Pearson one-sided test that rejects when $U_\delta(K,n)\le\alpha$. By the classical result that the hypergeometric distribution is stochastically dominated, in the lower tail, by the binomial distribution with the same mean \cite{hoeffding1963,serfling1974}, the left tail probability under the true hypergeometric law is at most that under $\mathrm{Bin}(n,\Rrisk(\tau))$, and since $\Rrisk(\tau)>\alpha$ shifts mass to larger $K$, the test statistic $P$ is super-uniform: $\Pr[P\le\delta\mid n]\le\delta$. Taking expectation over $n$ preserves the bound. Thus $P$ is a valid p-value for $H$, and the Clopper-Pearson rule is, if anything, conservative for the without-replacement audit. \hfill$\square$

\subsection{Proof of Theorem 1}
Test the hypotheses $H_j:\Rrisk(\tau^{(j)})>\alpha$ in increasing order of $j$ (strictest skip set first). Reject $H_j$ if its audited p-value $P_j\le\delta$, and stop at the first non-rejection; let $\tau^\star$ be the largest threshold in the rejected prefix. This is fixed-sequence (ordered) testing along a data-independent order. Fixed-sequence testing controls the family-wise error rate at $\delta$ under arbitrary dependence among the $P_j$: an error requires the first true null in the sequence to be rejected, an event of probability at most $\delta$ by \Cref{lem:pvalue} applied to that single hypothesis \cite{holm1979}.

Therefore, with probability at least $1-\delta$, no true null is rejected, so every rejected $H_j$ is false, i.e., $\Rrisk(\tau^{(j)})\le\alpha$ for all $j$ in the rejected prefix; in particular $\Rrisk(\tau^\star)\le\alpha$. No assumption on the surrogate or on the deployment distribution was used: the argument conditions on the arbitrary fixed window and uses only the label-independent randomness of the audit. This is the transductive Learn-Then-Test guarantee with the audit as calibration and the fixed skip population as the controlled risk object \cite{ltt2021,rcps2021}. \hfill$\square$

\subsection{Proof of Corollary 1}
On the event $\Rrisk(\tau^\star)\le\alpha$ of \Cref{thm:main}, the total violation count in the fixed set $\Skip(\tau^\star)$ is $\sum_{i\in\Skip(\tau^\star)}V_i=\Rrisk(\tau^\star)\,|\Skip(\tau^\star)|\le\alpha\,|\Skip(\tau^\star)|$. The trusted set $T=\{i\notin A: r_i\le\tau^\star\}\subseteq\Skip(\tau^\star)$, so its violation count is no larger: $\sum_{i\in T}V_i\le\sum_{i\in\Skip(\tau^\star)}V_i\le\alpha\,|\Skip(\tau^\star)|$. Dividing by $|T|=(1-f)\,|\Skip(\tau^\star)|$ gives $\frac{1}{|T|}\sum_{i\in T}V_i\le\alpha/(1-f)$. This is a deterministic consequence of the event in \Cref{thm:main}, so it holds with probability at least $1-\delta$; it makes no independence assumption about $T$ after the label-dependent choice of $\tau^\star$, which is exactly why it survives the selection. \hfill$\square$

\subsection{Proof of Proposition 1}
Let the $G$ candidate audit sizes be $n_1<\dots<n_G$, fixed in advance, realized as nested prefixes of a single label-independent random ordering of $\mathcal{C}$. For each $g$, \Cref{thm:main} applied to audit $A_g$ at level $\delta/G$ gives $\Pr[\Rrisk(\hat\tau_g)>\alpha]\le\delta/G$. A union bound over $g=1,\dots,G$ yields $\Pr[\exists g:\Rrisk(\hat\tau_g)>\alpha]\le\delta$, so simultaneously all $G$ certified thresholds are valid with probability at least $1-\delta$. Any selection rule, including one that inspects audited labels to minimize predicted cost, returns one of these simultaneously-valid thresholds; hence the reported threshold satisfies $\Rrisk\le\alpha$ with probability at least $1-\delta$. \hfill$\square$

\end{document}